\documentclass{article}

\usepackage{amsfonts,amsmath,amsthm,amssymb,enumerate,xcolor,siunitx,multirow,hhline,hyperref,breakurl}
\usepackage[margin=1in]{geometry}
\usepackage[makeroom]{cancel}
\usepackage[nottoc]{tocbibind}

\newtheorem{problem}{Problem}
\newtheorem{theorem}{Theorem}
\newtheorem{lemma}{Lemma}

\newtheorem{remark}{Remark}

\newcommand{\paren}[1]{\left(#1\right)}
\newcommand{\bracket}[1]{\left[#1\right]}
\renewcommand{\brace}[1]{\left\{#1\right\}}
\renewcommand{\ang}[1]{\left\langle#1\right\rangle}
\newcommand{\floor}[1]{\left\lfloor#1\right\rfloor}
\newcommand{\ceil}[1]{\left\lceil#1\right\rceil}

\newcommand{\F}{\mathbb{F}}

\newcommand{\ring}{\mathcal{R}}

\newcommand{\M}[1]{\begin{bmatrix}#1\end{bmatrix}}
\newcommand{\MA}[2]{\bracket{\begin{array}{#1}#2\end{array}}}
\newcommand{\diag}[1]{\mathrm{Diag}\paren{#1}}

\newcommand{\poly}{\mathrm{poly}}
\newcommand{\rk}[2][]{
    \ifthenelse{\equal{#1}{}}
    {\mathrm{rk}\paren{#2}}
    {\mathrm{rk}_{#1}\paren{#2}}
}

\newcommand{\rowspan}[1]{\mathrm{rowspan}\paren{#1}}

\renewcommand{\ker}[1]{\mathrm{ker}\paren{#1}}

\newcommand{\rref}[1]{\mathrm{rref}\paren{#1}}
\newcommand{\GL}[2]{\mathrm{GL}\paren{#1,\ #2}}

\renewcommand{\Vec}[1]{\mathrm{vec}\paren{#1}}
\renewcommand{\span}[1]{\mathrm{span}\paren{#1}}
\newcommand{\cpdeval}[1]{\bracket{\!\bracket{#1}\!}}
\newcommand{\unfold}[2]{#1_{\paren{#2}}}

\usepackage[noend]{algpseudocode}
\usepackage{algorithm,algorithmicx}

\newcommand*\Let[2]{\State #1 $\gets$ #2}
\algrenewcommand\algorithmicrequire{\textbf{Precondition:}}
\algrenewcommand\algorithmicensure{\textbf{Postcondition:}}

\newcommand{\Yield}[1]{\State \textbf{yield} #1}

\newcommand{\mainResult}{O^*(|\F|^{\min\brace{R,\ \sum_{d\ge 2} n_d} + (R-n_0)(\sum_{d\ne 0} n_d)})}
\newcommand{\secondResult}{O^*\paren{|\F|^{n_0+n_2 + (R-n_0+1-r_*)(n_1+n_2)+r_*^2}}}
\newcommand{\rStar}{\floor{\frac{R}{n_0}}+1}
\newcommand{\old}{O^*(|\F|^{n_0+(R-n_0)(\sum_d n_d)})}

\title{Faster search for tensor decomposition over finite fields}
\date{}
\author{Jason Yang}

\begin{document}

\maketitle

\begin{abstract}
We present an $O^*(|\mathbb{F}|^{\min\left\{R,\ \sum_{d\ge 2} n_d\right\} + (R-n_0)(\sum_{d\ne 0} n_d)})$-time algorithm for determining whether the rank of a concise tensor $T\in\mathbb{F}^{n_0\times\dots\times n_{D-1}}$ is $\le R$, assuming $n_0\ge\dots\ge n_{D-1}$ and $R\ge n_0$. For 3-dimensional tensors, we have a second algorithm running in $O^*(|\mathbb{F}|^{n_0+n_2 + (R-n_0+1-r_*)(n_1+n_2)+r_*^2})$ time, where $r_*:=\left\lfloor\frac{R}{n_0}\right\rfloor+1$. Both algorithms use polynomial space and improve on our previous work, which achieved running time $O^*(|\mathbb{F}|^{n_0+(R-n_0)(\sum_d n_d)})$.
\end{abstract}

\section{Introduction}
Given a tensor (multidimensional array) $T\in \ring^{n_0\times\dots\ n_{D-1}}$ over a ground ring $\ring$, a \textit{rank-$R$ canonical polyadic decomposition (CPD)} of $T$ is a list of matrices $A_d\in\ring^{n_d\times R},\ 0\le d<D$ such that \[T = \M{\sum_{0\le r<R} \prod_{0\le d<D} (A_d)_{i_d,r}}_{i_0,\dots,i_{D-1}}.\]

We call the $A_d$ ``factor matrices" and abbreviate the right-hand side as $\cpdeval{A_0,\dots,A_{D-1}}$. Note that this expression is also equal to $\sum_r \bigotimes_d (A_d)_{:,r}$, where $\otimes$ denotes the tensor product.

The rank of $T$, denoted $\rk{T}$, is the smallest $R$ such that there exists a rank-$R$ CPD of $T$. Determining tensor rank is the central problem underlying fast matrix multiplication \cite{blaser2013fast}. Formally, the action of multiplying a $m\times k$ matrix with a $k\times n$ matrix can be represented with a $mk\times kn\times nm$ tensor $\ang{m,k,n}$; then a rank-$R$ CPD of this tensor can be converted into a divide-and-conquer algorithm for multiplying two $N\times N$ matrices in $O(N^{3\log_{mkn} R})$ time. The quantity $3\log_{mkn} R$ is known as the \textit{running time exponent} of such a CPD. The famous Strassen algorithm \cite{strassen1969gaussian} corresponds to a rank-7 CPD of $\ang{2,2,2}$.
Furthermore, every fast matrix multiplication algorithm corresponds to a CPD of some $\ang{m,k,n}$, as long as it is restricted to arithmetic operations \cite{blaser2013fast}.

The asymptotically fastest known algorithm for matrix multiplication \cite{alman2025more} and its predecessors correspond to CPDs of very large $\ang{m,k,n}$ tensors; this is a consequence of applying several algebraic techniques, each of which produces \textit{sequences} of CPDs whose running time exponents converge to some limit. As a result, the constant factors of such algorithms render them impractical \cite{blaser2013fast}, despite new developments to mitigate this issue \cite{alman2025improving}.

An alternative approach to fast matrix multiplication is to directly find low-rank CPDs of small $\ang{m,k,n}$ tensors. Much work has been done in this direction using computer search \cite{smirnov2013bilinear, courtois2011new, heule2019local, fawzi2022discovering, kauers2023flip, kauers2025some, arai2024adaptive, deza2023fast}, which we detail in Section \ref{prior-work}.
We are most interested in the latter approach to fast matrix multiplication, due to impractical constant factors in the former approach. However, we deviate from previous research in two ways:
\begin{enumerate}[1.]
    \item We restrict ourselves to \textit{exact} algorithms which can prove existence or non-existence of a low-rank CPD of a given tensor. Although heuristic search methods such as \cite{courtois2011new, smirnov2013bilinear, heule2019local, fawzi2022discovering, kauers2023flip} have found many novel CPDs, they still have not resolved questions such as whether $\ang{3,3,3}$ has a rank-22 CPD, which has been open since 1976 \cite{laderman1976noncommutative}.
    
    A polynomial-time exact algorithm is unlikely to exist, as tensor rank is known to be NP-hard over a finite field and the rationals \cite{haastad1990tensor}, and undecidable over the integers \cite{shitov2016hard}.
    However, we are still interested in finding asymptotic improvements, as they can significantly expand the set of tensors that can be feasibly analyzed in practice, even if the final time complexity is exponential.

    To ensure an exact algorithm is possible in principle, we restrict ourselves to searching for CPDs over a finite field, so that the search space is finite. Doing so is not too strong of a handicap, as some previous work \cite{fawzi2022discovering, kauers2023flip, kauers2025some} has successfully obtained new algorithms for matrix multiplication by lifting CPDs from a finite field to the integers.

    \item We generalize to arbitrary tensors, not just those for matrix multiplication. Our reasons for doing so are: to find general search optimizations in tensor CPD that may have been overlooked; and to make our algorithms applicable to problems besides matrix multiplication that use CPD, such as the light bulb problem \cite{alman2023generalizations} and dynamic programming \cite{alman2023tensors}.
\end{enumerate}




\noindent Formally, we want to solve the following problem:
\begin{problem}
\label{cpd}
Given a tensor $T$ over a finite $\F$, and a rank threshold $R$, do one of the following:
\begin{itemize}
    \item If a rank-$R$ CPD of $T$ exists: return such a CPD;
    \item Else: determine that no such CPD exists.
\end{itemize}
\end{problem}

In a previous work of ours \cite{yang2024depth}, we proved that Problem \ref{cpd} can be solved in $\old$ time and $O^*(1)$ space, assuming without loss of generality (WLOG) that $R\ge n_0\ge \dots \ge n_{D-1}$ and $T$ is \textit{concise}; see Section \ref{conciseness} for more details. The $O^*$ notation ignores polynomial factors in all variables that appear in its expression: for both the time and space bounds of this result, $\poly(n_0,\dots,n_{D-1},D,R,|\F|)$ factors are ignored.

In this paper, we give two improvements on our previous time bound while maintaining polynomial space. The first applies to all tensors:

\begin{theorem}
\label{thm:main}
Given a concise tensor $T\in\F^{n_0\times\dots\times n_{D-1}}$, finding a rank-$R$ CPD of $T$ or determining that none exists can be done in $\mainResult$ time and $O^*(1)$ space.
\end{theorem}

The second improvement only applies to tensors of dimension $D=3$ --- see Section \ref{main2} for comparison against Theorem \ref{thm:main}:

\begin{theorem}
\label{thm:main2}
Given a concise tensor $T\in\F^{n_0\times n_1\times n_2}$, finding a rank-$R$ CPD of $T$ or determining that none exists can be done in $\secondResult$ time and $O^*(1)$ space, where $r^*:=\rStar$.
\end{theorem}

For both algorithms, the polynomial overhead for both time and space is approximately $R\prod_d n_d$, due to row reduction and computing tensor contractions (defined in Section \ref{notation}).

As an example of these improvements, consider $T=\ang{3,3,3}$, the tensor of multiplying two $3\times 3$ matrices. This tensor is the smallest matrix multiplication tensor involving square matrix shapes whose rank is not known; the best bounds on its rank over an arbitrary field are $\ge 19$ \cite{blaser2003complexity} and $\le 23$ \cite{laderman1976noncommutative}. Consider setting the rank threshold to $R=19$ and the searching over $\F_2$; then the exponential part of the time complexity of brute force is $2^{3\cdot 9\cdot 19}=2^{513}$, of our previous work \cite{yang2024depth} is $2^{9+(19-9)\cdot 27}=2^{279}$, of Theorem \ref{thm:main} is $2^{9+(19-9)\cdot 18}=2^{189}$, and of Theorem \ref{thm:main2} is $2^{9+9+(19-9+1-2)\cdot 18 + 2^2}=2^{184}$. Polynomial and constant factors are comparable across these algorithms.

We note that our results are incomparable with the following \textit{meet-in-the-middle} (MITM) technique: generate a set $S$ containing all tensors created by a rank-$\floor{\frac{R}{2}}$, then enumerate every rank-$\ceil{\frac{R}{2}}$ CPD $\cpdeval{B_0,\dots,B_{D-1}}$ and query whether $T-\cpdeval{B_0,\dots,B_{D-1}}\in S$. By implementing $S$ as any set data structure supporting logarithmic-time queries (e.g., a self-balancing binary tree), MITM runs in $O^*\paren{|\F|^{\ceil{\frac{R}{2}}\sum_d n_d}}$ time, which is faster than Theorems \ref{thm:main} and \ref{thm:main2} when $R$ is large; however, it requires $O^*\paren{|\F|^{\floor{\frac{R}{2}}\sum_d n_d}}$ space, rendering it impractical for all but the smallest tensors.

\subsection{Prior work}
\label{prior-work}
Below is a non-exhaustive list of existing methods to find CPDs of small matrix multiplication tensors $\ang{m,k,n}$:

\begin{itemize}
    \item By hand: Strassen's algorithm was extended by Hopcroft and Kerr \cite{hopcroft1971minimizing} to prove $\rk{\ang{m,2,n}}\le \ceil{\frac{3mn+\max\brace{m,n}}{2}}$.
    Laderman \cite{laderman1976noncommutative} proved $\rk{\ang{3,3,3}}\le 23$, which today is still the best known upper bound for $\ang{3,3,3}$. Many other upper bounds of small $\ang{m,k,n}$ are listed on Sedoglavic's table \cite{sedoglavic2019yet}.
    
    However, the most successful manual approach has been Pan's trilinear aggregation \cite{pan1978strassen, pan1982trilinear}, consisting of several upper bounds on $\rk{\ang{n,n,n}}$. The lowest running time exponent obtained by this approach is $\approx 2.7734$, from the rank bound $\rk{\ang{44,44,44}}\le 36133$.

    \item Numerical optimization: Smirnov \cite{smirnov2013bilinear} converts tensor CPD into an optimization problem whose objective is to minimize the sum-of-squares error between the target tensor $T$ and a CPD $\cpdeval{A,B,C}$ of fixed rank. The primary optimization technique is \textit{alternating least squares}, where one cycles through each factor matrix $A,B,C$, and analytically minimizes the objective with respect to that matrix while fixing all other factor matrices. Together with clever regularization, along with post-processing (to convert a sufficiently accurate CPD into an exact CPD), Smirnov proved $\rk{\ang{3,3,6}}\le 40$ over the rationals (running time exponent $\approx 2.7743$), nearly surpassing trilinear aggregation \cite{pan1982trilinear}.
    
    \item Boolean SAT: several works, such as Courtois et. al. \cite{courtois2011new} and Heule et. al. \cite{heule2019local}, have searched for CPDs over the finite field $\F_2$ by formulating the problem as boolean SAT, specifically to try finding a rank-22 CPD of $\ang{3,3,3}$. Although unsuccessful, thousands of new rank-23 CPDs were discovered (over $\F_2$) that are inequivalent to Laderman's CPD.

    \item Reinforcement learning: Fawzi et. al. \cite{fawzi2022discovering} (AlphaTensor) formulates tensor CPD as a single-player game where one starts with a given tensor, a move consists of subtracting away a rank-1 tensor, and the goal is to end with all-zeros in as few moves as possible.
    Using a neural network together with Monte-Carlo tree search, several new CPDs were found, including $\rk{\ang{3,4,5}}\le 47$ over the integers (running time exponent $\approx 2.8211$) and $\rk{\ang{4,4,4}}\le 47$ over $\F_2$ (running time exponent $\approx 2.7773$); the latter result could not be lifted to the integers.

    \item Flip graphs: Kauers and Moosbauer \cite{kauers2023flip, kauers2025some} begin with an arbitrary CPD $T=\sum_r a_r\otimes b_r\otimes c_r$, then (when possible) iteratively apply the following ``flip" transformation to a randomly chosen pair of summands with matching factors in the first axis,

    \[a\otimes b\otimes c + a\otimes b'\otimes c'\]
    \[\rightarrow a\otimes (b+b')\otimes c + a\otimes b'\otimes (-c+c'),\]

    or an equivalent transformation with the tensor axes and the order of the summands permuted.
    It is worth noting that this is essentially an elementary row operation in a matrix factorization.
    
    After each iteration, a CPD is also ``reduced" when possible (detected using some conditions involving linear span) so that its rank decreases by 1.
    

    In their first paper \cite{kauers2023flip}, Kauers and Moosbauer found a rank-95 CPD of $\ang{5,5,5}$ over $\F_2$ by initializing their search at the rank-96 CPD discovered by AlphaTensor \cite{fawzi2022discovering}.
    In a subsequent paper \cite{kauers2025some}, the authors proved $\rk{\ang{2,6,6}}\le 56$ over the integers by lifting a CPD from $\F_2$, making the first improvement on $\ang{2,6,6}$ since Hopcroft and Kerr \cite{hopcroft1971minimizing}.

    A variant of this search procedure, called ``adaptive flip graphs", was introduced by Arai et. al. \cite{arai2024adaptive}. This procedure occasionally \textit{increases} the rank of a CPD via the transformation
    \[a\otimes b\otimes c \rightarrow a'\otimes b\otimes c + (a-a')\otimes b\otimes c,\]

    which prevents the search from getting stuck at CPDs that either have no flip transformations or cannot be reduced. Along with extra search constraints, the authors of \cite{arai2024adaptive} found a rank-94 CPD of $\ang{5,5,5}$ over $\F_2$ (running time exponent $\approx 2.8229$).

    \item Constraint programming: closest to the spirit of our work, Deza et. al. \cite{deza2023fast} use the IBM ILOG CP Optimizer to exhaustively search for CPDs with elements in $\brace{-1,0,1}$, with symmetry-breaking constraints added that do not affect correctness (e.g., forcing lexicographic order among the summands $a_r\otimes b_r\otimes c_r$). Although the authors do not find new upper bounds on tensor rank, they do recover the previously known bounds $\rk{\ang{2,2,2}}\le 7$ and $\rk{\ang{2,2,3}}\le 11$.
\end{itemize}

Despite impressive progress on finding new CPDs for small tensors $\ang{m,k,n}$, ironically the best result is still Pan's latest trilinear aggregation result \cite{pan1982trilinear} from 1982. Furthermore, because none of the listed methods are exhaustive, except constraint programming, it is impossible to say whether the currently known best upper bounds really are optimal or whether the search methods just need to be run for longer.

Even Deza et. al.'s constraint programming \cite{deza2023fast} is not quite exhaustive, because of the rather artificial restriction that all elements in a CPD are in $\brace{-1,0,1}$ (with arithmetic over the integers). On the other hand, searching over a finite field has the advantage that if a CPD of a specific tensor with fixed rank is not found, that rules out the existence of such a CPD over the the entire ring of integers (since there exists a ring homomorphism from the integers to any finite field of prime order).

\subsection{Notation}
\label{notation}
\begin{itemize}
    \item Tensor slices are denoted with NumPy notation.
    \begin{itemize}
        \item e.g., for a two-dimensional tensor $A$: $A_{r,:}$ is the $r$-th row of $A$; $A_{:,:c}$ denotes the submatrix containing the first $c$ columns of $A$.
    \end{itemize}

    \item The tensor product of $A\in\ring^{n_0\times\dots\times n_{D-1}}$ and $B\in\ring^{s_0\times\dots\times s_{E-1}}$
    is
    \[A\otimes B:=\M{A_{i_0,\dots,i_{D-1}} B_{j_0,\dots,j_{E-1}}}_{i_0,\dots,i_{D-1}, j_0,\dots,j_{E-1}}.\]

    The chain tensor product $A_0\otimes\times\otimes A_{n-1}$ is denoted as $\bigotimes_{0\le i<n} A_i$.
    
    \item The axis-$d$ contraction of $T\in\ring^{n_0\times\dots\times n_{D-1}}$ by $M\in\ring^{n'_d \times n_d}$ is
    \[M\times_d T:=\M{\sum_{i_d} M_{i'_d,i_d} T_{_{i_0,\dots,i_{D-1}}}}_{i_0,\dots,i_{d-1}, i'_d, i_{d+1},\dots,i_{D-1}}.\]

    \item A rank-$R$ CPD is a list of factor matrices $A_d\in\ring^{n_d\times R},\ 0\le d<D$ that evaluates to the tensor $\cpdeval{A_0,\dots,A_{D-1}}:=\sum_r \bigotimes_d (A_d)_{:,r}$.
    \begin{itemize}
        \item It is straightforward to check that $M\times_d \cpdeval{A_0,\dots,A_{D-1}}
        \\=\cpdeval{A_0,\dots,A_{d-1},MA_d,A_{d+1},\dots,A_{D-1}}$.
    
        \item For brevity, we notate a CPD as ``$T=\cpdeval{A_0,\dots,A_{D-1}}$" instead of ``$(A_d)_d$ such that $T=\cpdeval{A_0,\dots,A_{D-1}}$".
    \end{itemize}

    \item $\rref{M}$ denotes the reduced row echelon form of matrix $M$, the unique matrix $M'$ such that:
    \begin{itemize}
        \item there exists some invertible $Q$ such that $M'=QM$
        \item the first $r$ rows of $M'$ are nonzero, where $r:=\rk{M}$, and all other rows are all-zero;
        \item the leftmost element of each nonzero row (the ``leading element" of that row) is a 1;
        \item the leading element of row $i$ is strictly to the right of the leading element of row $i-1$, for $1\le i<r$;
        \item each leading element is the only nonzero element in its column;
    \end{itemize}

    For example, the following matrix is in reduced row echelon form, where each $*$ can be replaced with an arbitrary value:
    \[\M{1&0&*&0&*\\0&1&*&0&*\\0&0&0&1&*\\0&0&0&0&0}\]

    \item $\GL{n}{\F}$ denotes the set of invertible $n\times n$ matrices with elements in the field $\F$.

    \item $\diag{v}$ denotes the diagonal matrix $\M{v_0\\&\ddots\\&&v_{n-1}}$.

    \item When writing tensor equations, we may ignore the presence of extra axes along which a tensor has length 1.
    \begin{itemize}
        \item e.g., if $M$ is a $m\times n$ matrix and $T$ is a $1\times m\times n$ tensor such that $T_{0,:,:}=M$, we may write $T=M$.
    \end{itemize}

\end{itemize}

\section{Conciseness}
\label{conciseness}
Before describing our algorithms, we first need to preprocess the input tensor $T\in\F^{n_0\times\dots\times n_{D-1}}$ so that it is \textit{concise}. Call $T$ \textit{axis-$d$ concise} if the set $\brace{T_{:,\dots,:,i_d,:,\dots,:}:0\le i_d<n_d}$ is linearly independent, where $i_d$ is indexing at axis $d$; and call $T$ \textit{concise} if it is axis-$d$ concise for all $0\le d<D$.

Construct the following objects for each axis $d$:
\begin{itemize}
    \item the \textit{axis-$d$ unfolding} of $T$, $\unfold{T}{d}:=\MA{c}{\vdots\\\hline \Vec{T_{:,\dots,:,i_d,:,\dots,:}}\\\hline \vdots}_{i_d}$, where $\Vec{}$ flattens a tensor in some consistent manner (e.g., row-major order);
    
    \item an arbitrary matrix $Q_d\in\GL{n_d}{\F}$ such that $Q_d \unfold{T}{d} = \rref{\unfold{T}{d}}$; such $Q_d$ can be constructed in polynomial time via row reduction;
    
    \item the axis-$d$ rank of $T$, $r_d:=\rk{\unfold{T}{d}}$;
\end{itemize}

Then define $T'=(Q_{D-1})_{:r_{D-1},:}\times_{D-1} \paren{\dots \paren{(Q_0)_{:r_0,:}\times_0 T}}$. It is clear from construction that $T'$ is concise.
Additionally, any CPD of $T=\cpdeval{A_0,\dots,A_{D-1}}$ can be converted into a CPD of $T'=\big[\!\big[(Q_0)_{:r_0,:}A_0,\dots,$ $(Q_{D-1})_{:r_{D-1},:}A_{D-1}\big]\!\big]$ consisting of the same rank. Finally, $T=(Q_{D-1}^{-1})_{:,:r_{D-1}}\times_{D-1} \paren{\dots \paren{(Q_0^{-1})_{:,:r_0}\times_0 T'}}$, so the reverse process is also possible. Thus, searching for a rank-$R$ CPD of $T'$ is equivalent to searching such a CPD for $T$.

The above process shows that if there exists a rank-$R$ CPD $T=\cpdeval{A_0,\dots,A_{D-1}}$, then $R\ge \max_d n_d$; this is because $T_{(d)}=A_d \MA{c}{\vdots \\\hline \Vec{\otimes_{d'\ne d} (A_{d'})_{:,r}} \\\hline \vdots}_{r}$, which has rank $r_d$ by definition, but also has rank $\le R$ as a consequence of the matrix product.

As a corollary, if the rank threshold $R$ were set to $1$, then the existence of a rank-1 CPD for $T$ implies $\forall d:r_d\le 1$. But then $\prod_d r_d\le 1$, so $T'$ clearly has rank $\le 1$, and thus so does $T$, yielding the following lemma:

\begin{lemma}
\label{rank1}
Given an arbitrary tensor $T\in\F^{n_0\times\dots\times n_{D-1}}$, returning a rank-1 CPD of $T$ or determining that no such CPD exists can be done in $O^*(1)$ time.
\end{lemma}

\begin{remark}
Lemma \ref{rank1} can be implemented in $O(D \prod_d n_d)$ time, by unfolding $T$ along each axis $d$ and using Gaussian elimination with early termination to detect if $T_{(d)}$ has rank $\le 1$.
\end{remark}

\section{Algorithms}
We first review our previous algorithm \cite{yang2024depth}, then discuss improvements.

Let $T\in\F^{n_0\times\dots\times n_{D-1}}$ be a concise tensor and $R$ the rank threshold. WLOG we can permute axes of $T$ so that $n_0\ge \dots \ge n_{D-1}$. For nontriviality, assume $\forall d: n_d\ge 2$.

Suppose $T$ has a rank-$R$ CPD $\cpdeval{A_0,\dots,A_{D-1}}$. By performing row reduction on $A_0$, there exists some $Q\in\GL{n_0}{\F}$ such that $QA_0=\rref{A_0}$. Since $T$ is concise, we must have $\rk{A_0}=n_0$, so we can force $\rref{A_0}=\MA{c|c}{I_{n_0} & \dots}$ after simultaneously permuting columns of all $A_d$, where the $\dots$ is arbitrary.

Split all factor matrices by the $n_0$-th column: denote $X_d:=(A_d)_{:,:n_0},\ Y_d:=(A_d)_{:,n_0:}$. Suppose all $Y_d$ were fixed; abbreviating $T_Y:=\cpdeval{Y_0,\dots,Y_{D-1}}$, we have

\[Q\times_0 \paren{T-T_Y}=\cpdeval{I_{n_0},X_1,\dots X_{D-1}}.\]

Because each row of $I_{n_0}$ contains exactly one nonzero element at a distinct column, this equation is equivalent to the condition
\[\forall i: \rk{(Q\times_0 (T-T_Y))_{i,:,\dots,:}}\le 1.\]

Since $(Q\times_0 (T-T_Y))_{i,:,\dots,:}=Q_{i,:}\times_0 (T-T_Y)$, and a square matrix is invertible if and only if its rows are linearly independent, $Q$ exists if and only if the set
\[B:=\brace{v\in\F^{1\times n_0} : \rk{v\times_0 (T-T_Y)}\le 1}\]
contains $n_0$ many linearly independent vectors.

By enumerating $v$ and using Lemma \ref{rank1}, $B$ can be constructed in $O^*(|\F|^{n_0})$ time. Then finding $n_0$ many linearly independent row vectors in $B$ (or determining that no such row vectors exist) can be done in $O(\poly(n_0,|B|))=O^*(|\F|^{n_0})$ time via row reduction, which solves $Q$.

Repeating the above procedure for each assignment of $(Y_d)_d$ makes the algorithm run in $\old$ time.
The polynomial overhead comes from computing tensor contractions, detecting if tensors have rank $\le 1$, and extracting linearly independent vectors in $B$.

There are further optimizations one can apply, such as forcing each $Y_d$ for $d\ne 0$ to have normalized columns and forcing the CPD summands $\otimes_d (Y_d)_{:,r}$ to be lexicographically ordered with respect to $r$. We omit these optimizations in our analysis for simplicity and to focus on improving the dominant exponential part of time complexity.
Additionally, it is not clear how to apply lexicographic ordering to our algorithm for Theorem \ref{thm:main2}, described in Section \ref{main2}.
However, we do apply these optimizations in our computational experiments (Section \ref{sec:experiments}).

\subsection{Main improvement}
\label{main}
Now we prove Theorem \ref{thm:main}. The main idea of the proof is to modify the previous algorithm to avoid fixing $Y_0$.

First, we expand $Q\times_0 (T-T_Y)=Q\times_0 T - Q\times_0 T_Y=Q\times_0 T - \cpdeval{QY_0,Y_1,\dots,Y_{D-1}}$ and substitute $C:=QY_0$ for a variable matrix $C\in\F^{n_0\times (R-n_0)}$. Because $Q$ is required to be invertible, any solution to this new equation containing $C$ can be converted into a solution to the old equation containing $Y_0$, and vice versa, so we do not need to place any constraints on $C$.

Solving for $Q$ is thus equivalent to having \[\forall i: \rk{Q_{i,:}\times_0 T - \cpdeval{C_{i,:},Y_1,\dots,Y_{D-1}}}\le 1.\] Since each inner condition only uses $C_{i,:}$ and not the entirety of $C$, such $Q$ exists if and only if the set 
\[S:=\Big\{v\in\F^{1\times n_0} : \exists c\in\F^{1\times (R-n_0)} \textrm{ s.t. }\]
\[\rk{v\times_0 T - \cpdeval{c,Y_1,\dots,Y_{D-1}}}\le 1\Big\}\]
contains $n_0$ many linearly independent vectors. If so, a rank-$R$ CPD of $T$ can be reconstructed in polynomial time, provided that for each $v$ we also keep track of a feasible $c$.

To solve for $Q$ efficiently, the obvious method is to construct $S$ directly by enumerating all possible pairs $(v,c)$; doing so takes $O^*(|\F|^{n_0} |\F|^{R-n_0})=O^*(|\F|^R)$ time.

A different method to solve $Q$ runs in $O^*(|\F|^{\sum_{d\ge 2} n_d})$ time; note that this summation omits $d=1$ as well as $d=0$.
The condition $\rk{v\times_0 T - \cpdeval{c,Y_1,\dots,Y_{D-1}}}\le 1$ is equivalent to
\[\exists u_d\in\F^{n_d} \textrm{ s.t. } v\times_0 T - \cpdeval{c,Y_1,\dots,Y_{D-1}}=\bigotimes_{d\ge 1} u_d.\]

If we fix $u_2,\dots,u_{D-1}$, this equation becomes a linear system over $(v,c,u_1)$. Solving for these remaining variables is equivalent to finding the nullspace of the linear map $f(v,c,u_1):=v\times_0 T - \cpdeval{c,Y_1,\dots,Y_{D-1}}-\bigotimes_{d\ge 1} u_d$. Using Gaussian elimination, we can find a set of basis solutions $P:=\brace{p_0,\dots,p_{k-1}}$, for some $k\le n_0$, such that $\span{P}=\ker{f}$.

Iterating over all possible assignments of $U:=(u_2,\dots,u_{D-1})$, we have $S=\bigcup_U \span{\brace{v:(v,c,u_1)\in P_U}}$, where $P_U$ denotes the set $P$ constructed from a specific assignment of $U$. Constructing $S$ directly is costly, but since we only need to extract a sufficient number of linearly independent vectors in $S$, it suffices to replace $S$ with a subset $S'\subseteq S$ such that $\span{S'}=\span{S}$. Setting $S':=\bigcup_U \brace{v:(v,c,u_1)\in P_U}$ satisfies these requirements, and such $S'$ takes $O^*(|U|n_0)=O^*\paren{|\F|^{\sum_{d\ge 2} n_d}}$ time to construct.


By using the faster of these two methods for solving $Q$ \footnote{This can be decided by pretending their asymptotic time complexities are exact measures of running time.}, we can solve $Q$ in $O^*\paren{|\F|^{\min\brace{R,\ \sum_{d\ge 2} n_d}}}$ time.
Repeating for each assignment of $(Y_d)_{d\ne 0}$ makes the overall running time \[\mainResult.\]
With small modifications, the space complexity can be improved to $O^*(1)$, as demonstrated in Algorithm \ref{alg:main}.

    \begin{algorithm}
        \caption{Search algorithm for Theorem \ref{thm:main}}
        \label{alg:main}
        \begin{algorithmic}[1]
            \Require{
                $T\in\F^{n_0\times\dots\times n_{D-1}}$, $R\ge 0$;
                $T$ concise, $n_0\ge \dots\ge n_{D-1}$
            }
            \Ensure{
                returns a rank-$\le R$ CPD of $T$, if one exists; else, returns null
            }
            \Function{search}{$T, R$}
                \If{$R<n_0$}
                    \State \Return{\textbf{null}}
                \EndIf
                \For{$1\le d<D,\ Y_d\in\F^{n_d\times (R-n_0)}$}
                    \Let{$\Delta$}{\Call{test}{$T,R,(Y_d)_{1\le d<D}$}}
                    \If{$\Delta$ \textrm{not null}}
                        \State \Return{$\Delta$}
                    \EndIf
                \EndFor
                \State \Return{\textbf{null}}
            \EndFunction
            
            \Function{test}{$T,R,(Y_d)_{1\le d<D}$}
                \Let{$(Q,C)$}{$([],[])$}
                \For{$(v,c)\in$ \Call{good\_pairs}{$T,R,(Y_d)_{1\le d<D}$}}
                    \If{$\rk{\MA{c}{Q\\\hline v}}>\rk{Q}$}
                        \Let{$(Q,C)$}{$\paren{\MA{c}{Q\\\hline v}, \MA{c}{C\\\hline c}}$}
                    \EndIf
                \EndFor
                \If{$Q$ has $n_0$ many rows}
                    \Let{$(X_d)_{1\le d<D}$}{$([])_{1\le d<D}$}
                    \For{$0\le i<n_0$}
                        \Let{$T_\square$}{$Q_{i,:}\times_0 T - \cpdeval{C_{i,:},Y_1,\dots,Y_{D-1}}$}
                        \Let{$(u_d)_{1\le d<D}$}{\Call{rank1}{$T_\square$}} \Comment{obtainable with Lemma \ref{rank1}}
                        \Let{$(X_d)_{1\le d<D}$}{$\paren{\MA{c|c}{X_d&u_d}}_{1\le d<D}$}
                    \EndFor
                    \State \textbf{return} $\Big(Q^{-1}\MA{c|c}{I_{n_0}&C}, \MA{c|c}{X_1&Y_1},\dots$ \\ $\MA{c|c}{X_{D-1}&Y_{D-1}}\Big)$
                \EndIf
                \State \Return{\textbf{null}}
            \EndFunction
            \Function{good\_pairs}{$T,R,(Y_d)_{1\le d<D}$}
                \If{$R \le \sum_{d\ge 2} n_d$}
                    \For{$v\in \F^{1\times n_0},\ c\in \F^{1\times (R-n_0)}$}
                        \If{$\rk{v\times_0 T - \cpdeval{c,Y_1,\dots,Y_{D-1}}}\le 1$}
                            \Yield{$(v,c)$} \Comment{emit output without storing in memory}
                        \EndIf
                    \EndFor
                \Else
                    \For{$2\le d<D,\ u_d\in\F^{n_d}$}
                        \State $f\gets\Big(\MA{c|c|c}{v&c&u_1}\mapsto v\times_0 T - \cpdeval{c,Y_1,\dots,Y_{D-1}}-u_1\otimes\paren{\bigotimes_{2\le d<D} u_d}\Big)$
                        \Let{$M$}{matrix s.t. $\rowspan{M}=\ker{f}$} \Comment{obtainable with Gaussian elimination}
                        \For{$0\le i<(\# \textrm{ rows in } M)$}
                            \Yield{$(M_{i,:n_0},M_{i,n_0:R})$}
                        \EndFor
                    \EndFor
                \EndIf
            \EndFunction
        \end{algorithmic}
    \end{algorithm}

\subsection{Extra improvement for 3-dimensional tensors}
\label{main2}

Here we prove Theorem \ref{thm:main2}.
When $D=3$, the tensor contractions $v\times_0 T$ are matrices and the CPDs $\cpdeval{c,Y_1,Y_2}$ are equivalent to matrix products $Y_1\diag{c}Y_2^\intercal$. Since matrix products are much better understood than general tensor CPDs, we can enumerate $Y_1,Y_2$ in a more clever way than brute force. All other steps remain the same as in the algorithm for Theorem \ref{thm:main}.

The following lemmas allow us to efficiently enumerate matrix factorizations of the form $M=UV$ for fixed $M$ and fixed shapes for $U,V$:


\begin{lemma}
\label{linsys}
Let $M\in\F^{m\times n}, U\in\F^{m\times k}$ be arbitrary matrices, and $s:=\rk{U}$.
Then the number of matrices $V\in\F^{k\times n}$ such that $M=UV$ is at most $|\F|^{(k-s)n}$, and all such matrices can be listed in $O^*\paren{|\F|^{(k-s)n}}$ time.
\end{lemma}
\begin{proof}
By row reduction, $\exists P\in\GL{m}{\F}, Q\in\GL{n}{\F}$ s.t. $PUQ=\MA{c|c}{I_s&O\\\hline O&O}$. Substituting $M'=PM$, $U'=PUQ$, and $V'=Q^{-1}V$, the condition $M=UV$ is equivalent to $M'=\MA{c|c}{I_s&O\\\hline O&O}V'=\MA{c}{V'_{:s,:}\\\hline O}$.

If $M'_{s:,:}\ne O$, there is no solution for $V'$; otherwise, the set of all solutions is $\brace{\MA{c}{M'_{:s,:}\\\hline L}: L\in\F^{(k-s)\times n}}$, and each element can be enumerated in polynomial time.
\end{proof}


\begin{lemma}
\label{rankfac}
Let $M\in\F^{m\times n}$ be an arbitrary matrix and $r:=\rk{M}$. For any $k\ge r$, the number of matrix pairs $(U\in\F^{m\times k},V\in\F^{k\times n})$ such that $M=UV$ is at most $|\F|^{(k-r)(m+n)+r^2}$, and all such pairs can be listed in $O^*\paren{|\F|^{(k-r)(m+n)+r^2}}$ time.
\end{lemma}
\begin{proof}
By row reduction, $\exists P\in\GL{m}{\F}, Q\in\GL{n}{\F}$ s.t. $PMQ=\MA{c|c}{I_r&O\\\hline O&O}$, so we can define $U':=PU, V'=VQ$ such that the condition $M=UV$ is equivalent to $\MA{c|c}{I_r&O\\\hline O&O}=U'V'$.

Split $U'=\MA{c}{U'_0\\\hline U'_1}$, where $U'_0:=U'_{:r,:},\ U'_1:=U'_{r:,:}$. Then the condition is equivalent to $U'_0V'=\MA{c|c}{I_r&O} \ \land \ U'_1V'=O$. A corollary of the first equation is that $\rk{U'_0}=r$.

Enumerate all possible $U'_0\in\F^{r\times k}$. For each $U'_0$ such that $\rk{U'_0}=r$, enumerate all $V'$ satisfying $U'_0V'=\MA{c|c}{I_r&O}$. Finally, for each $V'$, enumerate all $U'_1$ satisfying $U'_1V'=O \Leftrightarrow (V')^\intercal (U'_1)^\intercal=O$.

By Lemma \ref{linsys}, there are $O\paren{|\F|^{(k-r)n}}$ many $V'$ to enumerate for each fixed $U'_0$, since $\rk{U'_0}=r$; and $O\paren{|\F|^{(k-r)(m-r)}}$ many $U'_1$ to enumerate for each fixed $V'$, since $\rk{V'}\ge r$.
Thus, the total number of satisfying pairs $(U',V')$ is \[O\paren{|\F|^{rk}\cdot |\F|^{(k-r)n}\cdot |\F|^{(k-r)(m-r)}}\]
\[=O^*\paren{|\F|^{(k-r)(m+n)+r^2}},\] and enumerating each one takes polynomial time.
\end{proof}

Now we use these lemmas to solve the original problem.
As a reminder, we want to find some $Y_1,Y_2$ such that the set
\[S:=\brace{v\in\F^{1\times n_0} : \exists c\in\F^{1\times (R-n_0)} \textrm{ s.t. } \rk{v\times_0 T - \cpdeval{c,Y_1,Y_2}}\le 1}\]
contains $n_0$ many linearly independent vectors.

Since $n_0>0$ by the nontriviality conditions, we know $S$ must contain at least one vector, so we fix such a vector $v$. Then we enumerate $Y_1,Y_2$ such that $\exists c\in\F^{1\times (R-n_0)}$ s.t. $\rk{v\times_0 T - \cpdeval{c,Y_1,Y_2}}\le 1$. This condition is equivalent to $\exists c\in\F^{1\times (R-n_0)}, a\in\F^{n_1\times 1}, b\in\F^{1\times n_2} \textrm{ s.t. } v\times_0 T - \cpdeval{c,Y_1,Y_2}=ab$. Rearranging terms yields
\[v\times_0 T=\cpdeval{c,Y_1,Y_2}+ab\]
\[=\MA{c|c}{Y_1&a}\diag{\MA{c|c}{c&1}}\MA{c}{Y_2^\intercal \\\hline b}.\]

WLOG we can force $c$ to be a run of 0s followed by a run of 1s, since $S$ is invariant under simultaneous permutation of the columns of $Y_1$ and $Y_2$, and under independent scaling of columns of $Y_1, Y_2$.

Let $K:=R-n_0+1$ and $z$ be the number of 0s in $c$; then the right-hand side of the above equation is equal to $W_1W_2$,
where $W_1:=\MA{c|c}{Y_1&a}_{:,z:}\in\F^{n_1\times (K-z)}$ and $W_2:=\MA{c}{Y_2^\intercal \\\hline b}_{z:,:}\in\F^{(K-z)\times n_2}$. Then enumerate all $(W_1,W_2)$ such that $v\times_0 T=W_1 W_2$, and enumerate $(Y_1,Y_2)=\paren{\MA{c|c}{L_1&(W_1)_{:,:K-z-1}},\MA{c}{L_2\\\hline (W_2)_{:K-z-1,:}}^\intercal}$ over all arbitrary $L_1\in\F^{n_1\times z}, L_2\in\F^{z\times n_2}$.

By Lemma \ref{rankfac}, there are $O\paren{|\F|^{(K-z-r)(n_1+n_2)+r^2}}$ many pairs $(W_1,W_2)$ to enumerate.
From each pair, there are $O\paren{|\F|^{z(n_1+n_2)}}$
many $(Y_1,Y_2)$ to enumerate.
We have $z\le K-r$, otherwise $W_1 W_2$ would not have high enough rank. Summing over all possible $z$, the total cost of enumerating $Y_1,Y_2$ is
\[O^*\paren{\sum_{0\le z\le K-r} |\F|^{(K-z-r)(n_1+n_2)+r^2} |\F|^{z(n_1+n_2)}}\]
\[=O^*\paren{\cancelto{\textrm{polynomial}}{(K-r+1)}|\F|^{(K-r)(n_1+n_2)+r^2}}.\]

Since $v\times_0 T$ has shape $n_1\times n_2$, $r\le\min\brace{n_1,n_2}\le \frac{n_1+n_2}{2}$, so this cost is strictly decreasing with respect to $r$.

Finally, we iterate over $v$ to solve the original problem. However, we do not have to test all possible $v$.
Let $r_*$ be the smallest $r$ such that the set $S_r:=\brace{v\in\F^{1\times n_0}:\rk{v\times_0 T}\le r}$ contains $n_0$ many linearly independent vectors; then $S$ must in fact contain some $v$ such that $\rk{v\times_0 T}\ge r_*$, so we only need to iterate over such $v$.

We also notice that because $\rk{S_{r_*}}=n_0$, we can construct a rank-$(n_0r_*)$ CPD of $T$, since given $n_0$ many linearly independent vectors $v^{(0)},\dots,v^{(n_0-1)}\in S_{r_*}$, each axis-0 slice of $\MA{c}{v^{(0)}\\\hline \vdots \\\hline v^{(n_0-1)}}\times_0 T$ has rank $\le r_*$. Thus, if the rank threshold $R$ is $\ge n_0r_*$, we can solve the original CPD problem in only $O^*\paren{|\F|^{n_0}}$ time, as finding $r_*$ and constructing $S_{r_*}$ can be done within that time complexity.

Therefore, for worst-case analysis we only need to consider the case $r_*>\frac{R}{n_0}$. We also trivially have $r_*\le \min\brace{n_1,n_2}$ by considering matrix shape. WLOG assume $\frac{R}{n_0}<\min\brace{n_1,n_2}$, otherwise $R>n_0\min\brace{n_1,n_2}$ and a rank-$R$ CPD of $T$ would trivially exist.
Because the cost of enumerating $(Y_1,Y_2)$ for a given $v$ is strictly decreasing w.r.t. $r:=\rk{v\times_0 T}$, it suffices to set $r_*$ to the smallest integer strictly greater than $\frac{R}{n_0}$. Thus, the overall running time of our algorithm is
\[O^*\paren{\underbrace{|\F|^{n_0}}_{\textrm{enum } v} \cdot \underbrace{|\F|^{(R-n_0+1-r_*)(n_1+n_2)+r_*^2}}_{\textrm{enumerate } (Y_1,Y_2)} \cdot \underbrace{|\F|^{\cancelto{n_2}{\min\brace{R,\ n_2}}}}_{\textrm{construct } S}},\]
\[\textrm{ where } r^*:=\rStar.\]



To compare this algorithm (Theorem \ref{thm:main2}) against Theorem \ref{thm:main}, suppose all $n_d$ were $\Theta(n)$: then it is straightforward to check that Theorem \ref{thm:main2} is faster if $R\ge \omega(n)$, with an exponent in the $|\F|^{\dots}$ term of about $O(R)$ less than in Theorem \ref{thm:main}.
However, when Theorem \ref{thm:main} is augmented with forced lexicographic ordering, it becomes faster than Theorem \ref{thm:main2} for large enough $R$, since (for fixed $n$) $(R-\Theta(n))!=\frac{R!}{R^{\Theta(n)}}=|\F|^{\Theta(R\log_{|\F|} R)} \gg |\F|^{O(R)}$.
We suspect a similar lexicographic ordering can be applied to Theorem \ref{thm:main2}, but doing so is not obvious to us. For example, we cannot filter out all assignments of $(Y_1,Y_2)$ except those for which $(Y_1)_{:,r}\otimes (Y_2)_{:,r}$ is lexicographically nondecreasing with respect to $r$, because we assumed the columns of $Y_1,Y_2$ were already permuted such that $c=\M{0&\dots&0&1&\dots&1}$.


\section{Experiments}
\label{sec:experiments}
To test the practicality of our algorithms, we implemented Theorem \ref{thm:main} with several programmatic optimizations:
\begin{itemize}
    \item each column $(Y_d)_{:,r}$ is normalized so that its topmost nonzero element is 1 -- in particular, no $(Y_d)_{:,r}$ is assigned to all-zeros;
    \item the tuples $((Y_d)_{:,r})_d$ (called ``factor tuples") are strictly lexicographically increasing with respect to $r$;
    \item we allow all $Y_d$ to have $\le R-n_0$ many columns, not just $=R-n_0$ (but all $Y_d$ need the same number of columns);
\end{itemize}

The lexicographic ordering constraint reduces asymptotic (and real-world) time complexity by a factor of $\sim (R-n_0)!$, which is already significant for small inputs.

The last requirement allows the search algorithm to be correct with \textit{strict} lexicographic order, otherwise we would have to use \textit{nonstrict} order instead (e.g., a rank-$R$ CPD of a tensor with duplicate factor tuples could be reduced to a CPD with rank $<R$). Doing so does not change asymptotic complexity time or significantly increase constant factors, since there are exponentially many more possible assignments of $Y_d$ where the $Y_d$ have $R-n_0$ many columns than with $<R-n_0$ many columns.

We did not implement Theorem \ref{thm:main2} because it would be slower or have little speed improvement over our implementation of Theorem \ref{thm:main} with lexicographic ordering, as explained in Section \ref{main2}.

We tested our algorithm on two tensors. The first tensor was $\ang{2,2,2}$, which underlies Strassen's algorithm.
The second was the $4\times 4\times 4$ tensor $T$ such that
\begin{multline}
\label{eq:W2}
T_{0,:,:}=\M{1&\phantom{0}&\phantom{0}&\phantom{0}\\&\\&\\&},\ 
T_{1,:,:}=\M{&1&\phantom{0}&\phantom{0}\\1&\\&\\&},
T_{2,:,:}=\M{&\phantom{0}&1&\phantom{0}\\&\\1&\\&},\ 
T_{3,:,:}=\M{&&&1\\&&1\\&1\\1},
\end{multline}
originally from a presentation by Lysikov \cite{lysikov2024tensor}, who showed that this tensor has rank $\le 7$ over any field of characteristic $\ne 2$. Over $\F_2$, however, our algorithm found that $\rk{T}=8$.

We were not able to test the next largest nontrivial matrix multiplication tensor, $\ang{3,2,2}$, because its rank is 11 over an arbitrary field \cite{alekseyev1985complexity}, and ruling out rank $R=10$ over $\F_2$ would require searching $\sum_{0\le k\le 10-6} \binom{(2^6-1)(2^4-1)}{k}
\approx \num{3.32e10}$ many states.
Testing other tensors of shape larger than $4\times 4\times 4$ would also be infeasible, unless their rank is very close to the maximum side length along an axis.

For each tensor $T$ we tested, we first set $R$ to each of $0,1,2,\dots$ until a rank-$R$ CPD was found, which determines $\rk{T}$; then we re-ran the CPD search on the tensor $T':=Q_0\times_0 (\dots (Q_{D-1}\times_{D-1} T)\dots)$ for uniformly random invertible matrices $Q_d$, with rank threshold $R=\rk{T}$. Doing so ensures that the performance of our algorithm was not just due to finding a lucky assignment of the partial factor matrices $Y_d$ early in the search. We do not run this test for $R<\rk{T}$ because the number of visited states would be the exact same each time.
Experimental results are shown in Table \ref{tab:experiments}.

Below we present the minimal-rank CPDs that our algorithm found for each tensor, notated by their factor matrices such that $T=\cpdeval{A_0,A_1,\dots}$:
\begin{itemize}
    \item $T=\ang{2,2,2}$:
    \begin{itemize}
        \item $A_0=\M{
            0&0&1&0&0&1&0\\
            1&0&1&1&1&0&0\\
            0&1&0&1&0&1&1\\
            0&1&0&0&1&0&0\\
        }$
        \item $A_1=\M{
            1&0&1&1&0&0&1\\
            1&0&1&1&0&1&0\\
            1&1&0&1&0&0&1\\
            1&0&0&0&1&0&0\\
        }$
        \item $A_2=\M{
            0&0&1&1&0&1&1\\
            0&1&0&0&0&0&1\\
            1&0&0&1&0&1&1\\
            1&0&0&1&1&0&1\\
        }$
    \end{itemize}
    
    \item $T=(\textrm{tensor in Eq. } \ref{eq:W2})$:
    \begin{itemize}
        \item $A_0=\M{
        	1&0&0&0&0&0&0&0\\
        	1&0&0&1&0&0&1&0\\
        	1&0&1&0&0&1&0&0\\
        	1&1&0&0&1&1&1&1\\
        }$
        \item $A_1=\M{
        	1&1&1&1&0&0&0&0\\
        	0&0&0&1&0&0&1&1\\
        	0&0&1&0&0&1&0&1\\
        	0&1&0&0&1&0&0&0\\
        }$
        \item $A_2=\M{
        	1&1&1&1&0&0&0&0\\
        	0&0&0&1&0&0&1&1\\
        	0&0&1&0&0&1&0&1\\
        	0&1&0&0&1&0&0&0\\
        }$ \footnote{Notice that $A_2=A_1$. We are not sure why our algorithm happened to find a solution satisfying this property.}
    \end{itemize}
\end{itemize}

\begin{table}
    \centering
    \begin{tabular}{|l|l||l|l|l|l|}
        \hline
        $T$ & $R$ & \# trials & time (s) & \# states & statistic \\
        \hline
    \multirow{5}{4em}{$\ang{2,2,2}$} & 6 & 1 & 1.84 & 25426 & --- \\
        \hhline{~-----}
         & \multirow{4}{0.75em}{7} & \multirow{4}{1em}{100} & 1.22 & \num{1.91e+04} & geo mean \\
         & & & 0.202 & \num{3.28e3} & min \\
         & & & 1.83 & \num{2.90e4} & median \\
         & & & 8.52 & \num{1.33e5} & max \\
        \hline
    \multirow{5}{4em}{Eq. \ref{eq:W2}} & 7 & 1 & 123 & 1898626 & --- \\
        \hhline{~-----}
         & \multirow{4}{0.75em}{8} & \multirow{4}{1em}{10} & 26.3 & \num{3.49e5} & geo mean \\
         & & & 11.7 & \num{1.54e5} & min \\
         & & & 28.3 & \num{3.78e5} & median \\
         & & & 35.3 & \num{4.66e5} & max \\
        \hline
    \end{tabular}
    \caption{Experimental results of our algorithm on hand-selected tensors $T$, with rank thresholds $R=\rk{T}-1$ and $R=\rk{T}$. All results are over the field $\F_2$.
    Elapsed time was measured on a MacBook with a 2.3 GHz 8-Core Intel Core i9.
    \\
    For $R=\rk{T}$, we ran our algorithm on several random invertible axis-contractions of $T$, to mitigate lucky search. ``geo mean" means geometric mean.
    We do not run this test for any $R<\rk{T}$ because the algorithm would iterate over all possible assignments of $(Y_d)_d$ anyway and detect that no solution exists.
    }
    \label{tab:experiments}
\end{table}

We compare our algorithm against Deza et. al.'s constraint programming \cite{deza2023fast}, using geometric mean statistics:
\begin{itemize}
    \item For $T=\ang{2,2,2}$, we rule out $R=6$ in $1.84$ seconds and $\sum_{0\le k\le 6-4} \binom{(2^4-1)^2}{k}=24526$ states, whereas Deza et. al. takes $429$ seconds and $\num{3.28e8}$ branches (i.e., recursive calls in the constraint solver).
    For $R=7$, the stats are more comparable (1.22 sec, \num{1.91e4} states vs. 0.74 sec, \num{6.41e5} branches).

    \item Deza et. al. successfully found a rank-11 CPD for $\ang{2,3,2}$ (equivalent to $\ang{3,2,2}$ up to permutation of axes) in 26.47 seconds and \num{1.56e7} branches. Currently, this tensor is beyond the reach of our algorithm, though a faster implementation and processor could change that.

\end{itemize}

Our implementation is available at \url{https://github.com/coolcomputery/tensor-cpd-search}.

\section{Conclusion}
We have presented two exact algorithms for determining whether an arbitrary tensor over a finite field has rank at or below an arbitrary threshold. Both algorithms significantly improve over our previous result, and to our knowledge there is no other earlier work on this problem that guarantees exactness.

The second algorithm (Thm. \ref{thm:main2}) has a slight improvement in runtime complexity over the first (Thm. \ref{thm:main}), but (for large enough inputs) this improvement is smaller than forcing lexicographic ordering for the first algorithm. It may be possible to modify Thm. \ref{thm:main2} to support lex. ordering, so that its roughly $|\F|^{O(R)}$-factor theoretical time improvement becomes useful.

Theorem \ref{thm:main} can naturally be implemented as the following depth-first search: initialize $r\gets R-1$, fix the column tuple $((Y_d)_{:,r})_d$, then decrement $r$ and recurse. In fact, our implementation works in this manner.
It may be possible to prune some recursive calls by running fast (but possibly weak) conditions that rule out the existence of a rank-$R$ CPD given a partial assignment of the $Y_d$; we are actively investigating various such pruning methods. Orthogonally, one may adjust the order of assignments for the tuple $((Y_d)_{:,r})_d$ that are searched at each state, using heuristics and/or machine learning, so that if a rank-$R$ CPD exists it can be found early.
Neither pruning nor heuristic ordering break exactness. Although tight analysis of time complexity would become difficult, the actual speed on specific tensor inputs could become exponentially faster with suitable pruning and ordering.

\section*{Acknowledgments}
We thank Prof. Virginia Vassilevska Williams for her continual mentorship on this research topic throughout the last three years.

\end{document}